\documentclass[11pt,letterpaper]{article}
\usepackage[margin=1in]{geometry}

\usepackage{amsmath,amssymb,amsfonts}

\usepackage{amsthm}
\usepackage{bm}
\usepackage{mathrsfs}
\usepackage{xcolor}
\usepackage{graphicx}
\usepackage{cite}
\usepackage[hidelinks]{hyperref}

\usepackage{comment}

\newtheorem{assumption}{Assumption}
\newtheorem{remark}{Remark}
\newtheorem{definition}{Definition}
\newtheorem{lemma}{Lemma}
\newtheorem{theorem}{Theorem}

\title{\bfseries Control Barrier--Value Functions under Partial Observability:\\
Safety Guarantees via Conformal Prediction}

\author{Niloofar Jahanshahi\thanks{Department of Computing Science,
Simon Fraser University, Burnaby, BC, Canada.
\texttt{jahansha@sfu.ca}}
\and
Mo Chen\thanks{Department of Computing Science,
Simon Fraser University, Burnaby, BC, Canada.
\texttt{mochen@cs.sfu.ca}}}
\date{}
\begin{document}
\maketitle

\begin{center}
\small This work was supported by the Canada CIFAR AI Chairs program.
\end{center}

\begin{abstract}
This paper studies safety analysis and controller synthesis for partially observable nonlinear control systems. We extend the control barrier--value function (CBVF) framework, which combines Hamilton--Jacobi reachability and control barrier functions, to settings where full state information is not available and control is based on an estimated state. Given an estimator, we apply conformal prediction to the estimation error and obtain an error bound at a user-chosen miscoverage level. We incorporate this bound into the estimator-space safety analysis and define a CBVF-based safety certificate for partially observable systems. We then derive a finite-horizon probabilistic safety guarantee for the true system state. Finally, we propose a QP-based online safety filter for systems affine in the control and disturbance, whose solution enforces the CBVF safety condition in real time against bounded disturbance. The proposed framework is illustrated on a partially observable obstacle-avoidance case study.
\end{abstract}


\section{Introduction}

Control barrier--value functions (CBVFs) connect HJ reachability with CBF-based control \cite{ChoiCDC2021}. CBVFs preserve the main benefit of HJ analysis, namely a formally verified safety certificate obtained from a value function, while also introducing a barrier-style condition suitable for online control synthesis.
%%%%%%%
Although CBVFs offer an attractive combination of offline certification and online implementation, the existing framework largely assumes access to full state information. In many practical systems, however, only partial and noisy measurements are available; therefore, safety must be enforced using output information and an estimator. This challenge has been studied in related safety frameworks, in particular for CBFs and HJ reachability. On the CBF side, some works combine observers with barrier conditions to enforce safety using output measurements \cite{wang2022observercbf}. Other works develop CBF-based controller synthesis for partially observable systems through data-driven or estimator-based constructions \cite{jahanshahi2023po,jahanshahi2020partial}. On the HJ side, prior work has studied output-feedback safety-preserving control using state estimates together with uncertainty-aware safety analysis \cite{yousefi2017outputfeedback}, and has also considered how safety can be maintained during temporary loss of observability \cite{laine2020eyesclosed}. %These results show that partial observability can be incorporated into safety analysis, but the CBVF framework in \cite{ChoiCDC2021} has so far been limited to the full-state setting.

%%%%
In this paper, we propose a control barrier--value function framework for the safety analysis of nonlinear systems under partial observability. Because full state information is unavailable, our approach relies on an estimated state. To handle the resulting estimation error, we use conformal prediction (CP) to calibrate an error bound from data at a user-chosen miscoverage level \cite{ShaferVovkJMLR2008,AngelopoulosBatesFnT2023}. CP is well suited here because it yields finite-sample bounds without requiring an explicit probabilistic error model. By incorporating this calibrated bound into our analysis, we account for the mismatch between the estimated and true states, ensuring the safety certificate remains valid. Based on this uncertainty-aware description, we define a partially observable CBVF in estimator space and derive an online safety filter. Ultimately, our main contribution is this combined framework, which establishes a finite-horizon, high-probability safety guarantee for the true system state.
%%%%%%

%%%%%%%%%%%%%%%%%%%%%%%%%%%%%%%%%%%%%%%%%%%
\section{Background}
\label{sec:background}

\subsection{Problem Formulation: Dynamics and Safety Set}
\label{subsec:problem-formulation}
Consider a nonlinear control system over a backward horizon $[t,0]$, with $t\le 0$:
\begin{equation}
    \dot{x}(s) = f\big(x(s), u(s), d(s)\big),
    \qquad s \in [t,0],
    \label{eq:system-dynamics}
\end{equation}
where $x(s)\in\mathbb{R}^n$ is the state, $u(s)\in U\subset\mathbb{R}^m$ is the control, and $d(s)\in D\subset\mathbb{R}^{n_d}$ is an exogenous disturbance. We assume $U$ and $D$ are compact and convex.

Let $U[t,0]$ and $D[t,0]$ denote the sets of Lebesgue measurable signals with values in $U$ and $D$, respectively. As standard in differential games, the disturbance acts through a non-anticipative strategy $\xi_d:U[t,0]\to D[t,0]$: if two controls agree on $[t,\tau]$, then the induced disturbance signals also agree on $[t,\tau]$. Let $\Xi[t,0]$ denote the set of such strategies.

Safety is encoded by a function $l:\mathbb{R}^n\to\mathbb{R}$ with safe set
\begin{equation}
    \mathcal{L} := \{x\in\mathbb{R}^n : l(x)\ge 0\}.
    \label{eq:safe-set}
\end{equation}
The goal is to choose controls so that $x(s)\in\mathcal{L}$ for all $s\in[t,0]$ despite worst-case admissible disturbances.

\begin{assumption}[Standing regularity]
\label{ass:regularity}
The function $f:\mathbb{R}^n\times U\times D\to\mathbb{R}^n$ is continuous in $(x,u,d)$, uniformly Lipschitz in $x$ over $(u,d)\in U\times D$, and bounded on compact sets. The safety function $l$ is locally Lipschitz and bounded on compact subsets of $\mathbb{R}^n$.
\end{assumption}

%%%% HJ
\subsection{HJ Reachability Value Function}
\label{subsec:hj}
A standard robust safety characterization is given by HJ reachability \cite{TomlinLygerosSastry2000,MitchellBayenTomlin2005,BasarOlsder1999}. For initial condition $x(t)=x$, the safety payoff is defined as
\begin{equation}
    J\big(x,t,u(\cdot),\xi_d[u](\cdot)\big)
    := \min_{s \in [t,0]} l\big(x(s)\big),
    \label{eq:worst-safety-margin}
\end{equation}
and the corresponding value function is defined as
\begin{equation}
    V(x,t)
    :=
    \min_{\xi_d \in \Xi[t,0]}\ \max_{u(\cdot)\in U[t,0]}
    J\big(x,t,u(\cdot),\xi_d[u](\cdot)\big).
    \label{eq:hj-value-function}
\end{equation}
The zero-superlevel set of $V$ 
\begin{equation}
\mathcal{V}(t) := \{x\in\mathbb{R}^n : V(x,t)\ge 0\}
\label{eq:viability-kernel}
\end{equation}
is the finite-horizon viability kernel, i.e., the largest set of states from which safety can be maintained on $[t,0]$ against worst-case disturbances \cite{Aubin1991,MitchellBayenTomlin2005}. 

Under Assumption~\ref{ass:regularity}, $V$ is the viscosity solution of
\begin{equation}
0 = \min\Big( l(x) - V(x,t),\ \partial_t V(x,t) + H\big(x,\nabla_x V(x,t)\big) \Big),
\label{eq:hji-vi}
\end{equation}
with terminal condition $V(x,0)=l(x)$, where
\begin{equation}
H(x,p) := \max_{u\in U}\ \min_{d\in D}\ p\cdot f(x,u,d),
\label{eq:hamiltonian}
\end{equation}
with $p=\nabla_x V(x,t)$.

%%%% CBF
\subsection{Control Barrier Functions}
\label{subsec:cbf}
CBFs enforce safety through pointwise inequalities that can be checked online \cite{AmesTAC2017,AmesCBFSurvey2019}.

Compared with HJ reachability, CBFs are attractive online, but constructing a valid CBF for a general nonlinear system can be difficult. This gap motivates CBVFs, which combine reachability-based certification with barrier-style synthesis \cite{ChoiCDC2021}.

%%% CBVF
\subsection{Control Barrier--Value Functions}
\label{subsec:cbvf}
CBVFs modify the finite-horizon reachability objective so that the resulting value function satisfies a barrier-style variational inequality while verifying the same robust safe set as the HJ value function \cite{ChoiCDC2021}. For $\gamma\ge 0$, the CBVF is defined as
\begin{equation}
    B_\gamma(x,t)
    :=
    \min_{\xi_d \in \Xi[t,0]}\;
    \max_{u(\cdot)\in U[t,0]}\;
    \min_{s\in[t,0]}
    e^{\gamma (s-t)}\, l\big(x(s)\big),
    \label{eq:cbvf-def}
\end{equation}
where $x(\cdot)$ solves \eqref{eq:system-dynamics} with $x(t)=x$ and $d(\cdot)=\xi_d[u](\cdot)$. At $t=0$,
\begin{equation}
    B_\gamma(x,0)=l(x).
    \label{eq:cbvf-terminal}
\end{equation}

Under Assumption~\ref{ass:regularity}, $B_\gamma$ is the unique viscosity solution of
\begin{equation}
\begin{aligned}
0=\min\Big(&\, l(x)-B_\gamma(x,t),\\
&\, \partial_t B_\gamma(x,t)
    + H\big(x,\nabla_x B_\gamma(x,t)\big)
    + \gamma B_\gamma(x,t)\Big),
\end{aligned}
\label{eq:cbvf-vi}
\end{equation}
with terminal condition \eqref{eq:cbvf-terminal}, where the Hamiltonian $H(\cdot,\cdot)$ is defined as in~\eqref{eq:hamiltonian}.
Compared with \eqref{eq:hji-vi}, the extra $\gamma B_\gamma$ term gives a barrier-style condition that supports online filtering \cite{ChoiCDC2021}.

Motivated by \eqref{eq:cbvf-vi}, the safe control set is defined as
\begin{equation}
\begin{aligned}
K_{B_\gamma}(x,t)
\!\!:=\!\! \Big\{
\! u\in U\!\!: &\partial_t B_\gamma(x,t)
\!+ \!\min_{d\in D}\nabla_x B_\gamma(x,t)\!\cdot\! f(x,u,d)\\
&+ \gamma B_\gamma(x,t)\ge 0\Big\}.
\end{aligned}
\label{eq:kb-gamma}
\end{equation}
When $B_\gamma(x,t)\ge 0$, any $u\in K_{B_\gamma}(x,t)$ satisfies the barrier-style condition. For control-affine dynamics, \eqref{eq:kb-gamma} is affine in $u$, so it can be enforced through a quadratic program \cite{ChoiCDC2021,AmesTAC2017}.

\section{Partial Observability with Conformal Prediction Error Bounds}
\label{sec:po}

\subsection{Measurement Model}
\label{subsec:meas}

\textcolor{black}{
Although the plant dynamics $f$ in \eqref{eq:system-dynamics} are known,
the full state $x(s)$ is not available online because not all state
components can be measured. Instead, only partial information about $x(s)$ is available through the
measured output, defined as follows:}
\begin{equation}
    y(s) = h_y\big(x(s)\big) + v(s),
    \label{eq:output}
\end{equation}
where $h_y:\mathbb{R}^n\to\mathbb{R}^q$ is a known output map and $v(s)\in W\subset\mathbb{R}^q$ is bounded measurement noise.

\begin{assumption}[Output regularity and noise bound]
\label{ass:meas}
The map $h_y$ in \eqref{eq:output} is locally Lipschitz on the region of interest. The noise set $W$ is compact, and
\[
\bar v := \sup_{v\in W}\|v\| < \infty.
\]
Hence $\|v(s)\|\le \bar v$ for all $s$.
\end{assumption}

\subsection{Estimator Dynamics as an Analysis Model}
\label{subsec:estimator}

Let $\hat x(s)\in\mathbb{R}^n$ denote the estimator state. We model the estimator by
\begin{equation}
    \dot{\hat x}(s) = \hat f\big(\hat x(s),u(s),\hat d(s)\big),
    \qquad \hat d(s)\in \hat D(s),
    \label{eq:estimator-dynamics}
\end{equation}
where $\hat d$ collects the effect of measurement noise, plant disturbances as seen by the estimator, and modeling mismatch. \textcolor{black}{We assume that an estimator of the form
\eqref{eq:estimator-dynamics}, with known $\hat f$, is available.
The construction of the estimator is beyond the scope of this paper; several approaches
to estimator design for nonlinear systems have been developed in the
literature; see, e.g., the results in
\cite{Gauthier1992,KhalilPraly2014}.}

\begin{assumption}[Estimator-model regularity]
\label{ass:estimator}
The function $\hat f$ in \eqref{eq:estimator-dynamics} is continuous in $(\hat x,u,\hat d)$ and locally Lipschitz in $\hat x$, uniformly over $(u,\hat d)$ on the region of interest.
\end{assumption}

 The framework does not require a specific estimator structure; any estimator whose behavior can be conservatively captured by \eqref{eq:estimator-dynamics} is suitable.

%%%%%%% 

\subsection{Conformal Prediction Error Sets}
\label{subsec:cp}

Define the estimation error as $\varepsilon(s):=x(s)-\hat x(s)$. Let $\{t_k\}_{k=0}^N$ be sampling times on $[t,0]$, with $t_0=t$ and $t_N=0$, and write $\varepsilon_k:=\varepsilon(t_k)$.

\begin{assumption}[Exchangeability of rollouts]
\label{ass:exchangeability}
We have access to $M$ calibration rollouts. Rollout $j$ produces an error sequence $\{\varepsilon_k^{(j)}\}_{k=0}^N$ on the sampling grid. \textcolor{black}{
These $M$ calibration rollouts and a future online rollout, generated
under the same closed-loop data-generating process, are exchangeable as
random elements of $\mathbb{R}^{n\times(N+1)}$.
}
\end{assumption}

\textcolor{black}{
Although the true state is unavailable during online deployment, we assume that full-state information is available during offline calibration, allowing the estimation errors $\varepsilon_k^{(j)}=x_k^{(j)}-\hat x_k^{(j)}$ to be computed. Accordingly, the present framework is limited to settings in which such full-state calibration data are available offline.} 

To obtain a uniform bound over the sampling grid, we define the rollout-level score
\begin{equation}
    S^{(j)} := \max_{k=0,\dots,N}\big\|\varepsilon_k^{(j)}\big\|,
    \qquad j=1,\dots,M.
    \label{eq:rollout-score}
\end{equation}
Let $S_{(1)}\le \cdots \le S_{(M)}$ denote the sorted scores. For miscoverage level $\alpha\in(0,1)$, let
\begin{equation}
    q_\alpha := S_{(k^\star)}, \qquad
    k^\star := \left\lceil (M+1)(1-\alpha)\right\rceil,
    \label{eq:qalpha-rollout}
\end{equation}
which is the standard split-conformal quantile \cite{ShaferVovkJMLR2008,AngelopoulosBatesFnT2023}.
Under Assumption~\ref{ass:exchangeability}, the conformal guarantee is finite-sample for any calibration size $M$, but smaller $M$ generally yields a coarser and often more conservative estimate of $q_\alpha$.

For each sampling time $t_k$, define
\begin{equation}
    E_{\alpha,k} := \{\varepsilon\in\mathbb{R}^n:\ \|\varepsilon\|\le q_\alpha\}.
    \label{eq:Ealpha-def}
\end{equation}
Then, under Assumption~\ref{ass:exchangeability},
\begin{equation}
    \mathbb{P}\Big(\varepsilon_k\in E_{\alpha,k}\ \ \forall k=0,\dots,N\Big)\ge 1-\alpha.
    \label{eq:cp-coverage-grid}
\end{equation}
That is, with probability at least $1-\alpha$, the estimation error stays within radius $q_\alpha$ at all sampling times. For later use, we define
\begin{equation}
    r_\alpha := q_\alpha.
    \label{eq:ralpha}
\end{equation}

\begin{remark}[Error dependence]
\label{rem:error-dependence}
The estimation error $\varepsilon_k=x(t_k)-\hat x(t_k)$ may depend on the
state and trajectory history; independence within a rollout is not
required. The score in \eqref{eq:rollout-score} treats the complete error
trajectory as one sample. Thus, \eqref{eq:cp-coverage-grid} is uniform
over the sampling grid but marginal over exchangeable rollouts, rather
than conditional on a particular state or trajectory.
\end{remark}

%%%%%%%%%%%
\subsection{Continuous-Time Error Propagation Between Samples}
\label{subsec:ct-prop}

The guarantee \eqref{eq:cp-coverage-grid} is only on the sampling grid. To obtain continuous-time safety over $[t,0]$, we propagate the bound between samples.

\begin{assumption}[Error growth bound]
\label{ass:error-growth}
There exist constants $L_\varepsilon\ge 0$ and $\bar w\ge 0$ such that
\begin{equation}
    \|\dot\varepsilon(s)\| \le L_\varepsilon \|\varepsilon(s)\| + \bar w,
    \qquad \forall s\in[t,0].
    \label{eq:error-diff-ineq}
\end{equation}
Such a bound can be derived on a compact region from the Lipschitzness of the plant and estimator dynamics together with bounded disturbances and noise.
\end{assumption}

Under Assumption~\ref{ass:error-growth}, Gr\"onwall's inequality gives, for any $s\in[t_k,t_{k+1}]$,
\begin{equation}
\|\varepsilon(s)\|
\le
\|\varepsilon(t_k)\|\,e^{L_\varepsilon (s-t_k)}
+\frac{\bar w}{L_\varepsilon}\Big(e^{L_\varepsilon (s-t_k)}-1\Big),
\label{eq:gronwall-bound}
\end{equation}
with the usual interpretation that $\frac{\bar w}{L_\varepsilon}(e^{L_\varepsilon \tau}-1)=\bar w\,\tau$ when $L_\varepsilon=0$ \cite{KhalilNonlinearSystems}.

On the event $\varepsilon(t_k)\in E_{\alpha,k}$, we define the continuous-time error radius as
\begin{equation}
r_{\alpha}^{\mathrm{ct}}(s;t_k)
:=
r_{\alpha}e^{L_\varepsilon (s-t_k)}
+\frac{\bar w}{L_\varepsilon}\Big(e^{L_\varepsilon (s-t_k)}-1\Big),
 s\in[t_k,t_{k+1}],
\label{eq:rct}
\end{equation}
and the propagated error set
\begin{equation}
E_{\alpha}^{\mathrm{ct}}(s;t_k)
:=
\{\varepsilon\in\mathbb{R}^n:\ \|\varepsilon\|\le r_{\alpha}^{\mathrm{ct}}(s;t_k)\}.
\label{eq:Ect}
\end{equation}

\begin{remark}[Sampling interval and conservatism]
\label{rem:sampling-conservatism}
By \eqref{eq:rct}, the propagated radius is largest at $t_{k+1}$ and
increases with $L_\varepsilon$ and $\Delta t=t_{k+1}-t_k$ through
$e^{L_\varepsilon\Delta t}$; for $L_\varepsilon=0$, it grows linearly
with $\Delta t$. Thus, one should use a tight valid $L_\varepsilon$ and
the largest $\Delta t$ yielding an acceptable endpoint radius. Smaller
intervals tighten \eqref{eq:Ect} but increase sampling, estimation, and
calibration costs. Changing $\Delta t$ requires recalibration of
$r_\alpha$ on the new grid.
\end{remark}

%%%%%%
\subsection{Bounding the Lumped Disturbance Set \texorpdfstring{$\hat D(s)$}{Dhat(s)}}
\label{subsec:dhat}

We now convert the estimation error bound into a disturbance bound for the estimator-space dynamics.

\begin{assumption}[Lipschitz sensitivity]
\label{ass:lipschitz}
There exists $L_{\mathrm{est},y}>0$ such that, for all $s\in[t,0]$,
\begin{equation}
    \|\hat d(s)\|\le L_{\mathrm{est},y}\,\|\varepsilon_y(s)\|,
    \label{eq:est-lip}
\end{equation}
where
\[
\varepsilon_y(s):=y(s)-\hat y(s),
\qquad
\hat y(s):=h_y(\hat x(s)).
\]
\end{assumption}

The constant $L_{\mathrm{est},y}$ may be chosen as an offline upper bound on the gain from output error to the lumped disturbance in the estimator dynamics.

For $s\in[t_k,t_{k+1}]$, if $\varepsilon(s)\in E_{\alpha}^{\mathrm{ct}}(s;t_k)$ then $\|\varepsilon(s)\|\le r_{\alpha}^{\mathrm{ct}}(s;t_k)$. Since $h_y$ is Lipschitz on the region of interest,
\[
\|h_y(x(s))-h_y(\hat x(s))\|
\le L_{h_y}\|\varepsilon(s)\|
\le L_{h_y}r_{\alpha}^{\mathrm{ct}}(s;t_k),
\]
where $L_{h_y}$ is a Lipschitz constant of $h_y$. Using \eqref{eq:output}, this gives
\[
\|\varepsilon_y(s)\|
\le L_{h_y}r_{\alpha}^{\mathrm{ct}}(s;t_k)+\bar v.
\]
Combining with \eqref{eq:est-lip}, we obtain
\begin{equation}
    \|\hat d(s)\|
    \le
    L_{\mathrm{est},y}\big(L_{h_y}r_{\alpha}^{\mathrm{ct}}(s;t_k)+\bar v\big)
    =: r_{\hat d}(s;t_k).
    \label{eq:rdhat}
\end{equation}
We therefore define
\begin{equation}
    \hat D(s;t_k)
    :=
    \{\hat d\in\mathbb{R}^n:\ \|\hat d\|\le r_{\hat d}(s;t_k)\}.
    \label{eq:dhat-set}
\end{equation}
For notational convenience, let
\begin{equation}
\hat D(s):=\hat D(s;t_k),\qquad s\in[t_k,t_{k+1}].
\label{eq:dhat-piecewise}
\end{equation}

\begin{remark}[Estimator Bounds and Conservatism]
\label{rem:estimator-bounds}
The estimator model in \eqref{eq:estimator-dynamics}, the error-growth
bound in \eqref{eq:error-diff-ineq}, and the estimator-sensitivity bound
in \eqref{eq:est-lip} can be derived analytically
\cite{KhalilNonlinearSystems} or estimated conservatively from offline
data using bounded-error methods \cite{MilaneseVicino1991}. The
guarantee requires their validity. Larger bounds enlarge
\eqref{eq:Ect} and \eqref{eq:dhat-set}, tighten
\eqref{eq:po-tightened-safety}, and shrink the certified safe set.
\end{remark}

Let $\hat D[t,0]$ denote the set of measurable signals $\hat d:[t,0]\to\mathbb{R}^n$ such that $\hat d(s)\in\hat D(s)$ almost everywhere. As in the full-state case, the estimator disturbance acts through a non-anticipative strategy $\hat\xi_{\hat d}:U[t,0]\to\hat D[t,0]$. We denote the set of all such strategies by $\hat\Xi[t,0]$

%%%%%%%%
\subsection{Tightened Safety Function on Estimator State}
\label{subsec:tighten}

Given the propagated error set \eqref{eq:Ect}, we define the tightened safety margin
\begin{equation}
    \ell_\alpha(\hat x(s),s;t_k)
    :=
    \min_{\varepsilon \in E_{\alpha}^{\mathrm{ct}}(s;t_k)}
    l\big(\hat x(s)+\varepsilon\big),
    \qquad s\in[t_k,t_{k+1}],
    \label{eq:po-tightened-safety}
\end{equation}
where $l:\mathbb{R}^n\to\mathbb{R}$ is the safety function defining the safe set $\mathcal{L}$ in \eqref{eq:safe-set}. Thus, $\ell_\alpha(\hat x(s),s;t_k)\ge 0$ means that every state consistent with the propagated error tube around $\hat x(s)$ is safe at time $s$.

\begin{remark}[Obstacle inflation interpretation]
\label{rem:obstacle-inflation}
When the safety function $l$ encodes obstacle avoidance, the tightening in \eqref{eq:po-tightened-safety} effectively inflates the obstacle by the propagated error tube, enforcing safety against a conservative boundary that accounts for estimation uncertainty.
\end{remark}

\begin{lemma}[Correctness of tightening]
\label{lem:tightening}
Fix $s\in[t_k,t_{k+1}]$. If $x(s)=\hat x(s)+\varepsilon(s)$ with $\varepsilon(s)\in E_{\alpha}^{\mathrm{ct}}(s;t_k)$, then
\[
l(x(s))\ge \ell_\alpha(\hat x(s),s;t_k).
\]
\end{lemma}
\begin{proof}
By \eqref{eq:po-tightened-safety}, $\ell_\alpha(\hat x(s),s;t_k)$ is the minimum of $l(\hat x(s)+\varepsilon')$ over all $\varepsilon'\in E_{\alpha}^{\mathrm{ct}}(s;t_k)$. Evaluating this minimum at $\varepsilon'=\varepsilon(s)$ gives the result.
\hfill$\square$
\end{proof}

\begin{assumption}[Regularity of the tightened margin]
\label{ass:tightened}
For each interval $[t_k,t_{k+1}]$, the function $(\hat x,s)\mapsto \ell_\alpha(\hat x,s;t_k)$ is continuous on the region of interest.
\end{assumption}

%%%%%%%%%%%%%%%%%%%%%%%%%%%%%%%

\section{Partially Observable Control Barrier--Value Functions}
\label{sec:control}
In this section, we define the partially observable CBVF and its HJ characterization. Relative to the full-state CBVF \eqref{eq:cbvf-def}, this formulation introduces two changes: it operates in estimator space using \eqref{eq:estimator-dynamics}, and it replaces the nominal safety function $l$ with the tightened margin $\ell_\alpha$ from \eqref{eq:po-tightened-safety}. \textcolor{black}{
This estimator-space choice aligns the safety analysis with the
information available to the controller. The PO-CBVF is computed using
\eqref{eq:estimator-dynamics}, and the resulting safe control condition
is evaluated directly at $\hat x$, while \eqref{eq:Ect} and
\eqref{eq:dhat-piecewise} account for the difference between $\hat x$
and $x$. In contrast, a construction based on the original dynamics and
the state estimate would require enforcing the safety condition over all
true states consistent with $\hat x$ and the error set.
}

\subsection{PO-CBVF Definition}

Using the sampling grid $\{t_k\}_{k=0}^N$ from Section~\ref{sec:po}, we define the piecewise tightened margin
\begin{equation}
\ell_\alpha(\hat x(s),s) := \ell_\alpha(\hat x(s),s;t_k),
\qquad s\in[t_k,t_{k+1}],
\label{eq:ell-piecewise}
\end{equation}
where $\ell_\alpha(\hat x(s),s;t_k)$ is given by \eqref{eq:po-tightened-safety}. Since the propagated error bound is reinitialized at each sampling time, both $\ell_\alpha(\hat x,s)$ and $\hat D(s)$ may be piecewise continuous in time.

\begin{definition}[PO-CBVF]
\label{def:po-cbvf}
Fix $\gamma\ge 0$ and $\alpha\in(0,1)$. The partially observable control barrier--value function is
\begin{equation}
    B_{\gamma,\alpha}(\hat x,t)
    := \min_{\hat \xi_{\hat d} \in \hat \Xi[t,0]}
       \max_{u(\cdot)\in U[t,0]}
       \min_{s\in[t,0]} e^{\gamma (s-t)}\, \ell_\alpha\big(\hat x(s),s\big),
    \label{eq:po-cbvf-def}
\end{equation}
where $\hat x(\cdot)$ evolves according to \eqref{eq:estimator-dynamics} with $\hat d(\cdot)=\hat \xi_{\hat d}[u](\cdot)$.
\end{definition}

\begin{remark}[Reduction to the Standard CBVF]
\label{rem:reduction-standard-cbvf}
Note that if the estimation error bound in \eqref{eq:Ect} vanishes,
i.e., $\hat x=x$, and the estimator dynamics in
\eqref{eq:estimator-dynamics} are the same as the true dynamics in
\eqref{eq:system-dynamics}, then \eqref{eq:po-tightened-safety} gives
$\ell_\alpha(\hat x,s)=l(x(s))$. Therefore, the PO-CBVF in
\eqref{eq:po-cbvf-def} reduces to the standard CBVF in
\eqref{eq:cbvf-def}.
\end{remark}

The dynamic programming principle below leads to the HJI characterization.

\begin{theorem}[Dynamic programming principle for the PO-CBVF]
\label{thm:po-dpp}
Fix $t<t+\delta\le 0$ and suppose $[t,t+\delta]\subseteq [t_k,t_{k+1}]$ for some $k$. Then
{\small
\begin{equation}
\begin{aligned}
B_{\gamma,\alpha}(\hat x,t)
&=
\min_{\hat\xi_{\hat d}\in\hat\Xi[t,0]}
\max_{u(\cdot)\in U[t,0]}
\min\Big\{
\min_{s\in[t,t+\delta]}
e^{\gamma(s-t)}\,\ell_\alpha(\hat x(s),s),\\
&\hspace{3.5cm}
e^{\gamma\delta}
B_{\gamma,\alpha}(\hat x(t+\delta),t+\delta)
\Big\},
\end{aligned}
\label{eq:po-dpp}
\end{equation}
}%
where $\hat x(\cdot)$ evolves according to \eqref{eq:estimator-dynamics} with $\hat d(\cdot)=\hat\xi_{\hat d}[u](\cdot)$.
\end{theorem}

\noindent
The proof is a straightforward adaptation of the dynamic-programming argument in \cite[Theorem 2]{ChoiCDC2021}, and is omitted for brevity.

%%%%%%%%%%%%

\subsection{HJ Characterization}
Define the estimator-space Hamiltonian as
\[
H(\hat x,p,t) := \max_{u\in U}\min_{\hat d\in \hat D(t)} p\cdot \hat f(\hat x,u,\hat d),
\]
with $p=\nabla_{\hat x}B_{\gamma,\alpha}(\hat x,t)$. 

Adapting the viscosity-solution arguments of \cite{ChoiCDC2021} to the time-varying tightened margin and disturbance set, we find that under Assumptions~\ref{ass:estimator} and \ref{ass:tightened}, and the compactness of $U$ and $\hat D(t)$, the value function \eqref{eq:po-cbvf-def} satisfies the following variational inequality on each interval $[t_k,t_{k+1}]$:
\begin{equation}
\begin{aligned}
 0 = \min&\Big(\ell_\alpha(\hat x,t)-B_{\gamma,\alpha}(\hat x,t),\\
 &\partial_t B_{\gamma,\alpha}(\hat x,t)
\!+\! H(\hat x,\nabla_{\hat x} B_{\gamma,\alpha}(\hat x,t),t)
\!+\! \gamma B_{\gamma,\alpha}(\hat x,t)\Big),
\end{aligned}
\label{eq:po-cbvf-vi}
\end{equation}
with terminal condition
\begin{equation}
 B_{\gamma,\alpha}(\hat x,0) = \ell_\alpha(\hat x,0).
\label{eq:po-cbvf-terminal}
\end{equation}
The inner minimization of the Hamiltonian enforces robustness with respect to the estimator-space disturbance, whose admissible values are given by \eqref{eq:dhat-piecewise}. Across sampling times, consecutive interval problems are linked by the dynamic programming principle.

The corresponding estimator-space safe set is 
\begin{equation}
\mathcal{C}_{\gamma,\alpha}(t) := \{\hat x : B_{\gamma,\alpha}(\hat x,t)\ge 0\}.
\label{eq:chat-safe-set}
\end{equation}

\begin{remark}[Computational scalability]
\label{rem:scalability}
As an HJ-based method, the PO-CBVF inherits the curse of dimensionality
of standard grid-based HJ computations. Since
$x,\hat{x}\in\mathbb{R}^n$, formulating the problem in estimator space
does not increase the dimension of the underlying HJ problem relative
to the full-state CBVF. More scalable approaches to HJ reachability
have been proposed \cite{bansal2021deepreach}; their application to the
PO-CBVF is outside the scope of this paper.
\end{remark}

%%%%%%%%%%
\subsection{PO-CBVF-Induced Control Policy}

At points where $B_{\gamma,\alpha}$ is differentiable, \eqref{eq:po-cbvf-vi} defines the safe control set
\begin{equation}
\begin{aligned}
K_{B_{\gamma,\alpha}}(\hat x,t)
:= \Big\{u\in U : \,\, &\partial_t B_{\gamma,\alpha}(\hat x,t) \\
&+ \min_{\hat d\in \hat D(t)} \nabla_{\hat x} B_{\gamma,\alpha}(\hat x,t)\cdot \hat f(\hat x,u,\hat d)\\
&+ \gamma B_{\gamma,\alpha}(\hat x,t)\ge 0\Big\}.
\end{aligned}
\label{eq:po-kb}
\end{equation}
When $B_{\gamma,\alpha}(\hat x,t) < \ell_\alpha(\hat x,t)$, the second branch of \eqref{eq:po-cbvf-vi} is active and an optimal action can be chosen from
\[
\pi^\star(\hat x,t)\in \arg\max_{u\in U}\min_{\hat d\in \hat D(t)}
\nabla_{\hat x}B_{\gamma,\alpha}(\hat x,t)\cdot \hat f(\hat x,u,\hat d).
\]
When $B_{\gamma,\alpha}(\hat x,t)=\ell_\alpha(\hat x,t)$, any $u\in K_{B_{\gamma,\alpha}}(\hat x,t)$ preserves the barrier-style inequality against all admissible $\hat d\in\hat D(t)$. At nonsmooth points, the same interpretation can be stated using generalized derivatives consistent with the viscosity framework \cite{ChoiCDC2021,CrandallIshiiLions1992}.

%%%%%%
\subsection{High-Probability True-State Safety}
\label{subsec:true-safety}

We now state the main safety guarantee for PO-CBVF filtering.

\begin{theorem}[Finite-horizon high-probability safety of the true state under PO-CBVF filtering]
\label{thm:true-safety}
Consider the partially observable system \eqref{eq:system-dynamics} with measurements \eqref{eq:output} and estimator dynamics \eqref{eq:estimator-dynamics}. Let $B_{\gamma,\alpha}$ be defined by \eqref{eq:po-cbvf-def}, let $\mathcal{C}_{\gamma,\alpha}(t)$ be given by \eqref{eq:chat-safe-set}, and let $K_{B_{\gamma,\alpha}}$ be given by \eqref{eq:po-kb}. Assume $x(t)\in\mathcal{L}$ and $\hat x(t)\in\mathcal{C}_{\gamma,\alpha}(t)$. If
\[
u(s)\in K_{B_{\gamma,\alpha}}(\hat x(s),s)
\quad \text{for a.e. } s\in[t,0],
\]
then the resulting closed-loop true-state trajectory satisfies
\[
\mathbb{P}\big(x(s)\in\mathcal{L}\ \forall s\in[t,0]\big)\ge 1-\alpha.
\]
\end{theorem}

\begin{proof}
Let $\mathcal{E}_{\mathrm{grid}} := \{\varepsilon(t_k)\in E_{\alpha,k}\ \forall k=0,\dots,N\}$. By \eqref{eq:cp-coverage-grid}, $\mathbb{P}(\mathcal{E}_{\mathrm{grid}})\ge 1-\alpha$.

Since $u(s)\in K_{B_{\gamma,\alpha}}(\hat x(s),s)$ for a.e.\ $s\in[t,0]$ and $B_{\gamma,\alpha}(\hat x(t),t)\ge 0$, the barrier-style invariance argument implied by \eqref{eq:po-kb} gives $B_{\gamma,\alpha}(\hat x(s),s)\ge 0$ for all $s\in[t,0]$. Then \eqref{eq:po-cbvf-vi} implies $\ell_\alpha(\hat x(s),s)\ge 0$ for all $s\in[t,0]$.

On $\mathcal{E}_{\mathrm{grid}}$, Assumption~\ref{ass:error-growth} and \eqref{eq:gronwall-bound} imply $\varepsilon(s)\in E_{\alpha}^{\mathrm{ct}}(s;t_k)$ for each $s\in[t_k,t_{k+1}]$. Hence, Lemma~\ref{lem:tightening} gives
\[
l(x(s)) = l(\hat x(s)+\varepsilon(s)) \ge \ell_\alpha(\hat x(s),s;t_k) = \ell_\alpha(\hat x(s),s) \ge 0
\]
for all $s\in[t,0]$. Therefore, $x(s)\in\mathcal{L}$ for all $s\in[t,0]$ on $\mathcal{E}_{\mathrm{grid}}$, which implies $\mathbb{P}(x(s)\in\mathcal{L}\ \forall s\in[t,0])\ge \mathbb{P}(\mathcal{E}_{\mathrm{grid}})\ge 1-\alpha$.
\hfill$\square$
\end{proof}

\begin{remark}[Obstacle-avoidance interpretation]
\label{rem:true-safety-obstacle}
In the obstacle-avoidance setting, the safe set $\mathcal{L}$ is the complement of the physical obstacle region. The tightened margin $\ell_\alpha$ in \eqref{eq:po-tightened-safety} accounts for obstacle inflation induced by the propagated conformal prediction error bound, so that safety is enforced conservatively in estimator space. Therefore, Theorem~\ref{thm:true-safety} implies that, with probability at least $1-\alpha$, the true state avoids the physical obstacle over the entire horizon $[t,0]$. Equivalently, the probability that the true state enters the physical obstacle at some time in $[t,0]$ is at most $\alpha$.
\end{remark}

%%%%%%%%%%%%
\subsection{QP-Based Safety Control Synthesis}
\label{subsec:po-cbvf-qp}

For estimator dynamics affine in $u$ and $\hat d$,
\begin{equation}
\dot{\hat x} = \hat p(\hat x) + \hat q(\hat x)\,u + \hat r(\hat x)\,\hat d,
\label{eq:est-affine}
\end{equation}
a least-restrictive safe input is obtained by filtering a nominal input $u_{\mathrm{nom}}(\hat x,s)$ through the following quadratic program:
\begin{equation}
\begin{aligned}
\min_{u \in U}\quad & \|u - u_{\mathrm{nom}}(\hat x,s)\|^2 \\
\text{s.t.}\quad
& \partial_t B_{\gamma,\alpha}(\hat x,s)
+ \min_{\hat d \in \hat D(s)}
\nabla_{\hat x} B_{\gamma,\alpha}(\hat x,s)\cdot  \\
& \big(\hat p(\hat x)+\hat q(\hat x)u+\hat r(\hat x)\hat d\big)
+ \gamma B_{\gamma,\alpha}(\hat x,s)\ge 0.
\end{aligned}
\label{eq:po-cbvf-qp}
\end{equation}
This enforces $u(s)\in K_{B_{\gamma,\alpha}}(\hat x(s),s)$ whenever $B_{\gamma,\alpha}$ is differentiable. If $\hat D(s)$ has convenient structure, such as a polytope or norm ball, the inner minimization can often be evaluated in closed form.

\begin{remark}[Offline and online computations]
\label{rem:offline-online}
Note that, online, measurements and the applied input are used to update
$\hat x$, evaluate the precomputed PO-CBVF, and solve
\eqref{eq:po-cbvf-qp}; the true state $x$ is not required. Offline, the
scores and radius in \eqref{eq:rollout-score}--\eqref{eq:qalpha-rollout}
are computed using stored full states $x_k^{(j)}$ from calibration
rollouts and are used to construct the error and disturbance sets and
the PO-CBVF.
\end{remark}

\begin{remark}[QP Feasibility]
\label{rem:qp-feasibility}
At differentiable points in the PO-CBVF safe set
\eqref{eq:chat-safe-set}, the HJ computation accounts for $U$, making
the safe control set \eqref{eq:po-kb} nonempty and
\eqref{eq:po-cbvf-qp} feasible. Enlarging the error and disturbance sets
in \eqref{eq:Ect} and \eqref{eq:dhat-set} tightens
\eqref{eq:po-tightened-safety} and shrinks this safe set; outside it,
feasibility and safety are not guaranteed. Infeasibility is detected
from the solver status or constraint residual. A penalized slack can
restore feasibility, but a positive slack removes the guarantee of
Theorem~\ref{thm:true-safety}.
\end{remark}

%%%%%%%%%%%%%%

\section{Numerical Results}
\label{sec:numerical}
We evaluate the proposed PO-CBVF on a partially observable obstacle-avoidance problem for a planar double integrator with state $x=(p_x,p_y,v_x,v_y)$:
\begin{equation}
\dot p_x = v_x,\quad
\dot p_y = v_y,\quad
\dot v_x = u_x + d_x,\quad
\dot v_y = u_y + d_y,
\label{eq:num-true-dynamics}
\end{equation}
with control bound $u\in[-4,4]^2$, disturbance bound $d\in[-0.1,0.1]^2$, and circular obstacles of radius $2.0$ located in the state space.

Only noisy position measurements are available online:
\[
y_x = p_x + n_x,\qquad y_y = p_y + n_y,
\]
with $n_x,n_y\in[-0.1,0.1]$. The online estimator is
\begin{equation}
\begin{aligned}
\dot{\hat p}_x &= \hat v_x + L_{\mathrm{gain}}(y_x-\hat p_x),
\dot{\hat p}_y = \hat v_y + L_{\mathrm{gain}}(y_y-\hat p_y),\\
\dot{\hat v}_x &= u_x,
\dot{\hat v}_y = u_y,
\end{aligned}
\label{eq:num-estimator}
\end{equation}
with $L_{\mathrm{gain}}=2.0$. For offline PO-CBVF computation, we use the estimator-space analysis model
\begin{equation}
\dot{\hat p}_x = \hat v_x + \hat d_x,\qquad
\dot{\hat p}_y = \hat v_y + \hat d_y,\qquad
\dot{\hat v}_x = u_x,\qquad
\dot{\hat v}_y = u_y,
\label{eq:num-analysis-model}
\end{equation}
where the lumped disturbance $\hat d$ is bounded using the conformal error propagation in Section~\ref{sec:po}.

The bounds are calibrated from $M=500$ offline rollout trajectories with miscoverage level $\alpha=0.05$. To match the deployment setting, these calibration rollouts are generated on-policy using the same PD nominal controller and QP-based PO-CBVF safety filter used online, with process and measurement noises sampled from bounded uniform distributions. Using the rollout-level score in \eqref{eq:rollout-score} yields the conformal radius $q_\alpha=0.0742$. \textcolor{black}{
Since the radius affects the safety filter, calibration is repeated until
$q_\alpha$ stabilizes. In this case study, the control input cancels
between \eqref{eq:num-true-dynamics} and \eqref{eq:num-estimator}, so the
iteration does not alter the error-score distribution.
} The propagated error bound is then computed using $
L_\varepsilon=0,\qquad \bar w=0.0146,$
and is used to construct the effective obstacle boundary and the estimator-space disturbance bound $\hat D(t)$, with
$
L_{h_y}=1.0,\qquad L_{\mathrm{est},y}=2.0,\qquad \bar v=0.1315.$

The PO-CBVF is computed offline over horizon $T=5$ s with $\Delta t=0.02$ s and $\gamma=1.0$. Online, a PD nominal controller defined on the estimated state $\hat x$ is filtered through the QP in \eqref{eq:po-cbvf-qp}.

Figure~\ref{fig:spatial_trajectory} shows a representative closed-loop rollout. The system passes through the gap between the inflated obstacles and reaches the goal region safely. 
Empirically, safety exceeded the nominal $95\%$ guarantee, resulting from conservatism in the worst-case disturbance model and continuous-time error propagation.
\begin{figure}[htbp]
    \centering
    \IfFileExists{po_spatial_subplots.png}{%
      \includegraphics[width=\linewidth]{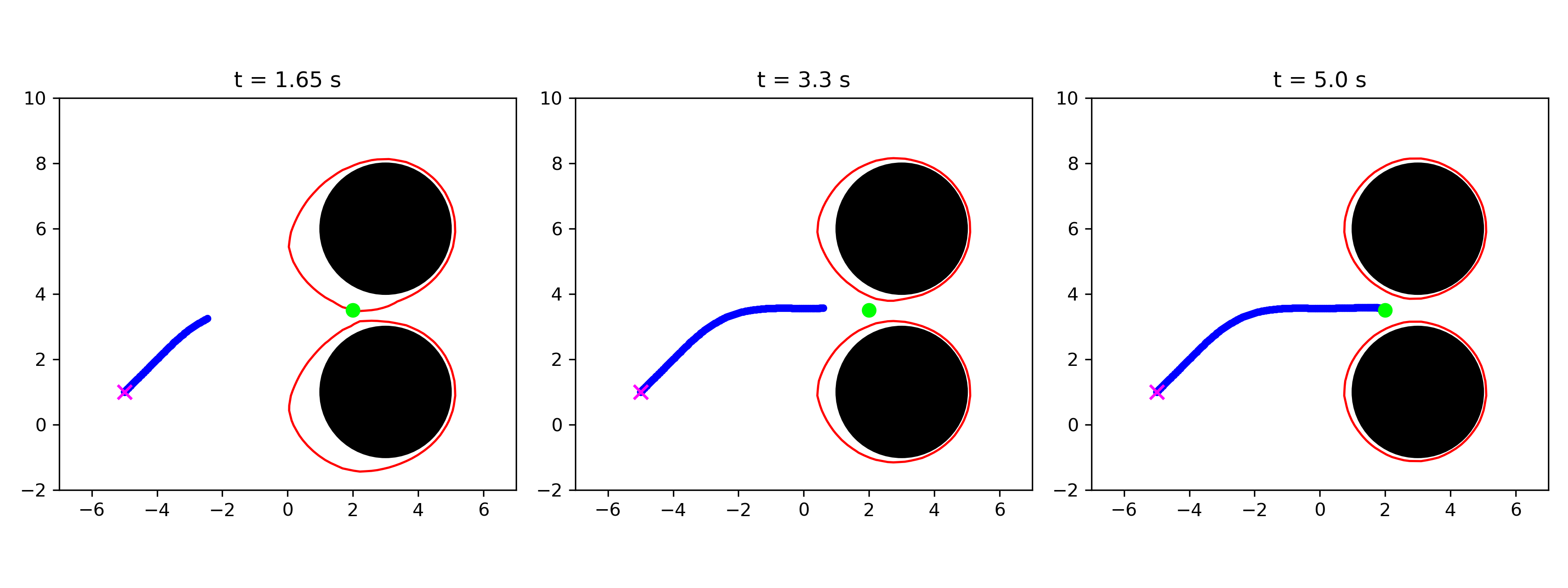}%
    }{%
      \fbox{\parbox[c][0.28\textheight][c]{0.9\linewidth}{\centering
      Add \texttt{po\_spatial\_subplots.png} to the arXiv source package.}}%
    }
 \caption{Representative true-state trajectory under PO-CBVF filtering at three time instants. The red curve shows the zero level set of the 2D spatial slice of $B_{\gamma,\alpha}$ obtained at the current estimated velocity.}
    \label{fig:spatial_trajectory}
\end{figure}

%%%%%%%%%%%%%%%%%%%%%%%

\section{Conclusion}
\label{sec:conclusion}

In this paper, we proposed a control barrier--value function framework for partially observable nonlinear systems. From this construction, we derived a CBVF-based safety certificate for partially observable systems, its Hamilton--Jacobi characterization, and an online QP-based safety filter. This framework establishes a finite-horizon, high-probability safety guarantee for the true system state. We demonstrated the approach on an obstacle-avoidance example, where true-system safety was achieved using only noisy output measurements.

%%%%%%%%%%%%%%%%%%%%%%%%% REFERENCES 
%\printbibliography
\bibliographystyle{unsrt}
\bibliography{ref}

@book{BasarOlsder1999,
  author    = {Ba\c{s}ar, Tamer and Olsder, Geert Jan},
  title     = {Dynamic Noncooperative Game Theory},
  edition   = {2},
  series    = {Classics in Applied Mathematics},
  volume    = {23},
  publisher = {SIAM},
  address   = {Philadelphia, PA},
  year      = {1999},
  doi       = {10.1137/1.9781611971132}
}

@book{Aubin1991,
  author    = {Aubin, Jean-Pierre},
  title     = {Viability Theory},
  publisher = {Birkh{\"a}user},
  address   = {Boston, MA},
  year      = {1991}
}

@inproceedings{wang2022observercbf,
  author    = {Wang, Yujie and Xu, Xiangru},
  title     = {Observer-based Control Barrier Functions for Safety Critical Systems},
  booktitle = {2022 American Control Conference (ACC)},
  pages     = {709--714},
  year      = {2022},
  doi       = {10.23919/ACC53348.2022.9867262}
}

@article{jahanshahi2023po,
  author  = {Jahanshahi, Niloofar and Zamani, Majid},
  title   = {Synthesis of Controllers for Partially-Observable Systems: A Data-Driven Approach},
  journal = {IFAC-PapersOnLine},
  volume  = {56},
  number  = {2},
  pages   = {5525--5530},
  year    = {2023},
  doi     = {10.1016/j.ifacol.2023.10.391}
}

@article{jahanshahi2020partial,
  author  = {Jahanshahi, Niloofar and Jagtap, Pushpak and Zamani, Majid},
  title   = {Synthesis of Partially Observed Jump-Diffusion Systems via Control Barrier Functions},
  journal = {IEEE Control Systems Letters},
  volume  = {5},
  number  = {1},
  pages   = {253--258},
  year    = {2021},
  doi     = {10.1109/LCSYS.2020.3001562}
}

@inproceedings{yousefi2017outputfeedback,
  author    = {Yousefi, Mahdi and van Heusden, Klaske and Mitchell, Ian M. and Dumont, Guy Albert},
  title     = {Output-Feedback Safety-Preserving Control},
  booktitle = {2017 American Control Conference (ACC)},
  pages     = {2550--2555},
  year      = {2017},
  doi       = {10.23919/ACC.2017.7963336}
}

@inproceedings{laine2020eyesclosed,
  author    = {Laine, Forrest and Chiu, Chih-Yuan and Tomlin, Claire},
  title     = {Eyes-Closed Safety Kernels: Safety of Autonomous Systems Under Loss of Observability},
  booktitle = {Robotics: Science and Systems XVI},
  year      = {2020},
  doi       = {10.15607/RSS.2020.XVI.096}
}

@book{KhalilNonlinearSystems,
  author    = {Khalil, Hassan K.},
  title     = {Nonlinear Systems},
  edition   = {3},
  publisher = {Prentice Hall},
  address   = {Upper Saddle River, NJ},
  year      = {2002}
}

@article{CrandallIshiiLions1992,
  author  = {Crandall, Michael G. and Ishii, Hitoshi and Lions, Pierre-Louis},
  title   = {User's Guide to Viscosity Solutions of Second Order Partial Differential Equations},
  journal = {Bulletin of the American Mathematical Society},
  volume  = {27},
  number  = {1},
  pages   = {1--67},
  year    = {1992},
  doi     = {10.1090/S0273-0979-1992-00266-5}
}

@article{TomlinLygerosSastry2000,
  author  = {Tomlin, Claire J. and Lygeros, John and Sastry, S. Shankar},
  title   = {A Game Theoretic Approach to Controller Design for Hybrid Systems},
  journal = {Proceedings of the IEEE},
  volume  = {88},
  number  = {7},
  pages   = {949--970},
  year    = {2000},
  doi     = {10.1109/5.871303}
}

@article{MitchellBayenTomlin2005,
  author  = {Mitchell, Ian M. and Bayen, Alexandre M. and Tomlin, Claire J.},
  title   = {A Time-Dependent Hamilton--Jacobi Formulation of Reachable Sets for Continuous Dynamic Games},
  journal = {IEEE Transactions on Automatic Control},
  volume  = {50},
  number  = {7},
  pages   = {947--957},
  year    = {2005},
  doi     = {10.1109/TAC.2005.851439}
}

@article{AmesTAC2017,
  author  = {Ames, Aaron D. and Xu, Xiangru and Grizzle, Jessy W. and Tabuada, Paulo},
  title   = {Control Barrier Function Based Quadratic Programs for Safety Critical Systems},
  journal = {IEEE Transactions on Automatic Control},
  volume  = {62},
  number  = {8},
  pages   = {3861--3876},
  year    = {2017},
  doi     = {10.1109/TAC.2016.2638961}
}

@inproceedings{AmesCBFSurvey2019,
  author    = {Ames, Aaron D. and Coogan, Samuel and Egerstedt, Magnus and Notomista, Gennaro and Sreenath, Koushil and Tabuada, Paulo},
  title     = {Control Barrier Functions: Theory and Applications},
  booktitle = {2019 18th European Control Conference (ECC)},
  pages     = {3420--3431},
  year      = {2019},
  doi       = {10.23919/ECC.2019.8796030}
}

@inproceedings{ChoiCDC2021,
  author    = {Choi, Jason J. and Lee, Donggun and Sreenath, Koushil and Tomlin, Claire J. and Herbert, Sylvia L.},
  title     = {Robust Control Barrier--Value Functions for Safety-Critical Control},
  booktitle = {2021 60th IEEE Conference on Decision and Control (CDC)},
  pages     = {6814--6821},
  year      = {2021},
  doi       = {10.1109/CDC45484.2021.9683085}
}

@article{ShaferVovkJMLR2008,
  author  = {Shafer, Glenn and Vovk, Vladimir},
  title   = {A Tutorial on Conformal Prediction},
  journal = {Journal of Machine Learning Research},
  volume  = {9},
  pages   = {371--421},
  year    = {2008}
}

@article{AngelopoulosBatesFnT2023,
  author  = {Angelopoulos, Anastasios N. and Bates, Stephen},
  title   = {Conformal Prediction: A Gentle Introduction},
  journal = {Foundations and Trends in Machine Learning},
  volume  = {16},
  number  = {4},
  pages   = {494--591},
  year    = {2023},
  doi     = {10.1561/2200000101}
}

@article{MilaneseVicino1991,
  author  = {Mario Milanese and Antonio Vicino},
  title   = {Estimation Theory for Nonlinear Models and Set Membership
             Uncertainty},
  journal = {Automatica},
  volume  = {27},
  number  = {2},
  pages   = {403--408},
  year    = {1991},
  doi     = {10.1016/0005-1098(91)90090-O}
}

@article{Gauthier1992,
  author  = {Jean-Paul Gauthier and Hassan Hammouri and Sami Othman},
  title   = {A Simple Observer for Nonlinear Systems: Applications to
             Bioreactors},
  journal = {IEEE Transactions on Automatic Control},
  volume  = {37},
  number  = {6},
  pages   = {875--880},
  year    = {1992},
  doi     = {10.1109/9.256352}
}

@article{KhalilPraly2014,
  author  = {Hassan K. Khalil and Laurent Praly},
  title   = {High-Gain Observers in Nonlinear Feedback Control},
  journal = {International Journal of Robust and Nonlinear Control},
  volume  = {24},
  number  = {6},
  pages   = {993--1015},
  year    = {2014},
  doi     = {10.1002/rnc.3051}
}

@inproceedings{bansal2021deepreach,
  author    = {Somil Bansal and Claire J. Tomlin},
  title     = {{DeepReach}: A Deep Learning Approach to High-Dimensional Reachability},
  booktitle = {2021 IEEE International Conference on Robotics and Automation (ICRA)},
  pages     = {1817--1824},
  year      = {2021},
  doi       = {10.1109/ICRA48506.2021.9561949}
}

\end{document}